\documentclass[a4paper, 10pt, conference,onecolumn]{ieeeconf}      % Use this line for a4
\IEEEoverridecommandlockouts                              % This command is only
\newtheorem{theorem}{Theorem}

\newtheorem{proposition}[theorem]{Proposition}
\newtheorem{assumption}[theorem]{Assumption}

\newtheorem{example}[theorem]{Example}
\def\R{{\Bbb R}}
\def\N{{\Bbb N}}

\def\E{{\Bbb E}}

\def\L{{\Bbb L}}
\def\X{\mathcal{X}}
\def\H{\mathcal{H}}

\def\A{\mathcal{A}}

\newcommand{\until}[1]{\{1,\dots, #1\}}

\usepackage{graphicx,subfigure}% for pdf, bitmapped graphics files
\usepackage{epsfig} % for postscript graphics files
\usepackage{mathptmx} % assumes new font selection scheme installed
\usepackage{times} % assumes new font selection scheme installed
\usepackage{amsmath} % assumes amsmath package installed
\usepackage{amssymb}  % assumes amsmath package installed
\usepackage{euscript}
\usepackage{color}
\usepackage{algorithmic}
\usepackage{algorithm}

\title{\bf Multi-agents adaptive estimation and coverage control\\ using Gaussian regression 
\thanks{ This work is supported by the European Community's Seventh Framework Programme [FP7/2007-2013] under
 grant agreement n. 257462 HYCON2 Network of excellence and by the MIUR FIRB project RBFR12M3AC-Learning meets time: a new computational approach to learning in dynamic systems .}}
\author{Andrea Carron, Marco Todescato, Ruggero Carli, Luca Schenato, Gianluigi Pillonetto
  \thanks{ A. Carron, M. Todescato, R. Carli, L. Schenato, M. Todescato, G. Pillonetto are with the
    Department of Information Engineering, University of Padova, 
    Via Gradenigo 6/a, 35131 Padova, Italy 
    {\tt\small \{carronan|todescat|carlirug|schenato|giapi\}
    @dei.unipd.it}. 
  }}
\date{}

\begin{document}

\maketitle

\begin{abstract}
We consider a scenario where the aim of a group of agents is to perform the optimal coverage
of a region according to a sensory function. In particular, centroidal Voronoi partitions have to be computed. 
The difficulty of the task is that the sensory function is unknown and has to be reconstructed on line from noisy measurements.
Hence, estimation and coverage needs to be performed at the same time.  We cast the problem in a Bayesian regression framework, where the sensory function  is seen as a Gaussian random field. Then, we design a set of control inputs which try to well balance coverage and estimation, also discussing
convergence properties of the algorithm. Numerical experiments show the effectivness of the new approach.
\end{abstract}

\section{Introduction}

The continuous progress on hardware and software is allowing the appearance of compact and relatively inexpensive autonomous vehicles embedded with multiple sensors (inertial systems, cameras, radars, environmental monitoring sensors), high-bandwidth wireless communication and powerful computational resources. While previously limited to military applications, nowadays the use of cooperating vehicles for autonomous monitoring and large environment, even for civilian applications, is becoming a reality. Although robotics research has obtained tremendous achievements with single vehicles, the trend of adopting multiple vehicles that cooperate to achieve a common goal is still very challenging and open problem. 

In particular, an area that has attracted considerable attention for its practical relevance is the problem of \emph{environmental partitioning problem} and \emph{coverage control} whose objective is to partition an area of interest into subregions each monitored by a different robot trying to optimize some global cost function that measures the quality of service provided by the monitoring robots.

The "centering and partitioning" algorithm originally proposed by Lloyd \cite{LLOYD:82} and elegantly reviewed in the survey \cite{DU:99} is a classic approach to environmental partitioning problems and coverage control problems. The Lloyd algorithm computes Centroidal Voronoi partitions as optimal configurations of an important class of objective functions called coverage functions. 
The Lloyd approach was first adapted for distributed coverage in the robotic multiagent literature
control in \cite{Cortes:04}; see also the text \cite{Bullo:09} (Chapter $5$ and literatures notes in Section $5.4$) for a comprehensive treatment. 
Since this beginning, similar algorithms have
been applied to non-convex environments \cite{Zhong08}, \cite{Pimenta08}, to dynamic routing with equitable partitioning \cite{Baron:07}, to robotic networks with limited anisotropic sensory \cite{Laventall:09} and to
coverage with communication constraints \cite{Bullo:12}.

Most of the works cited above assume that a global sensory cost function is known a priori by each agent. Therefore, the focus is limited to the distributed coverage control problem. However, it is often unrealistic to assume such function to be known. For instance, consider a group of underwater vehicles whose main goal is to monitor areas which present a higher concentration of pollution. The distribution of pollution is not known in advance, but vehicles are provided with sensors that can take noisy measurements of it. In this context, coverage control is much harder since the vehicles has to simultaneously explore the environment to estimate pollution distribution and to move to areas with higher pollution concentrations. This is a classical robotic task often referred to as \emph{coverage-estimation} problem. In \cite{Schwager}, an adaptive strategy is proposed to solve it but the agents are assumed to take an uncountable number of noiseless measurements. Moreover, the authors used a parametric approach with the assumption that the true function belongs to such class. More recently, \cite{Choi} proposed a non parametric approach based on Markov Random Fields for adaptive sampling and function estimation. This approach has the advantage to provide better approximation of the underlying sensory function as well confidence bounds on the estimate.

The novelty of this work is to consider a Bayesian non parametric learning scheme where, under the framework of Gaussian regression 
\cite{Rasmussen}, the unknown function is modeled as a zero-mean Gaussian random field. Robot coordination control is guaranteed to incrementally improve the estimate of the sensory function and simultaneously achieve asymptotic optimal coverage control. Although robot motion is generated by a centralized station, this work provides a starting point to design coordination algorithm for simultaneous estimation and coverage.
Note however that the \textit{robot to base station communication model} adopted in this paper already finds application
for ocean gliders interfaces communicating with a tower \cite{PereiraCommunication}, UAV data mules that periodically visit ground robots \cite{Shah}, or cost-mindful use of satellite or cellular communication.

Classical learning problem consists of estimating a function from examples collected on input locations drawn from a fixed probability density function (pdf) \cite{Poggio90,Smale2007}. Recent extensions also replace such pdf with a convergent sequence of probability measures \cite{SmaleOnline}. 
When performing coverage, the stochastic mechanism underlying the input locations establishes how the agents move inside the domain
of interest. The peculiarity of our algorithm is that such pdf is allowed to vary over time, depending also on the current estimate of the function. 
Hence, agents locations consist of a non Markovian process, leading to a learning problem where stochastic adaption may happen infinitely often
(with no guarantee of convergence to a limiting pdf). Under this complex scenario, we will derive conditions that ensure statistical consistency of 
the function estimator both assuming that the Bayesian prior is correct and relaxing this assumption. In this latter case, 
we assume that the function belongs to a suitable reproducing kernel Hilbert space and provide 
a non trivial extension of the statistical learning estimates derived in \cite{Smale2007} (technical details are gathered in Appendix).

The paper is so organized. After giving some mathematical preliminaries in Section \ref{MP}, problem statement is reported in Section \ref{PF}. The proposed algorithm is presented in Section \ref{AL}, with its convergence propriety discussed in Section \ref{CP}. In Section \ref{sct:NmrRsl} are reported some simulations results. Conclusions then end the paper.

\section{Mathematical preliminaries}
\label{MP}
Let $\mathcal{X}$ be a compact and convex polygon in $\R^2$ an let $\|\cdot\|$ denote the Euclidean distance function. Let $\mu: \mathcal{X} \to \R_{>0}$ be a distribution density function defined over $\mathcal{X}$. Within the context of this paper, a \emph{partition} of $\mathcal{X}$ is a collection of $N$ polygons $\mathcal{W}=(W_1,\ldots, W_N)$ with disjoint interiors whose union is $\mathcal{X}$. Given the list of $N$ points in $\mathcal{X}$, $\mathbf{x}=(x_1,\ldots,x_N)$, we define the Voronoi partition $\mathcal{V}(\mathbf{x})=\left\{V_1(\mathbf{x}),\ldots,V_N(\mathbf{x})\right\}$ generated by $\mathbf{x}$ as
$$
V_i(\mathbf{x})=\left\{q \in \mathcal{X} \,|\,\,\|q-x_i\|\leq \|q-x_j\|, \,\,\,\forall j \neq i\right\}.
$$
For each region $V_i$, $i \in \until{N}$, we define its centroid with the respect to the density function $\mu$ as
$$
c_i(V_i(\mathbf{x}))= \left(\int_{V_i(\mathbf{x})}\mu(q)dq \right)^{-1} \int_{V_i(\mathbf{x})} q \mu(q) dq.
$$
We denote by
$$
\textbf{c}(\mathcal{V}(\mathbf{x}))=(c_1(V_1(\mathbf{x})), \ldots, c_N(V_N(\mathbf{x})))
$$
the vector of regions centroids corresponding to the Voronoi partition generated by $\mathbf{x}=(x_1,\ldots,x_N)$. A partition is said to be a Centroidal Voronoi partition of the pair $\left( \mathcal{X}, \mu \right)$ if , for $i \in \until{N}$, the point $x_i$ is the centroid of $V_i(\mathbf{x})$. 

Given $\mathbf{x}=(x_1,\ldots,x_N)$ and a density function $\mu$ we introduce the \emph{Coverage function} $H(\mathbf{x}; \mu)$ defined as
$$
H(\mathbf{x}; \mu)= \sum_{i=1}^N \int_{V_i(\mathbf{x})}\|q-c_i(V_i(\mathbf{x}))\|^2 \mu(q) dq  
$$
For a fixed density function $\mu$, it can be shown that the set of local minima $H(\mathbf{x}; \mu)$ is composed by the  points $\mathbf{x}=(x_1,\ldots,x_N)$ are such $x_1,\ldots,x_N$ are the centroids of the corresponding regions $V_1(\mathbf{x}), \ldots, V_N(\mathbf{x})$, i.e, $\mathcal{V}(\mathbf{x})$ is a Centroidal Voronoi partition.

\subsection{Coverage Control Algorithm}
\label{CCA}
Let $\mathcal{X}$ be a convex and closed polygon in $\R^2$ and let $\mu$ be a density function defined over $\mathcal{X}$.  Consider the following optimization problem 
$$
\min_{\mathbf{x} \in Q^N} H(\mathbf{x}; \mu).
$$
The \emph{coverage algorithm} we consider is a version of the classic Lloyd algorithm based on "centering and partitioning" for the computation of Centroidal Voronoi partitions. Given an initial condition $\mathbf{x}(0)$ the algorithm cycles iteratively the following two steps:
\begin{enumerate}
\item computing the Voronoi partition corresponding to the current value of $\mathbf{x}$, namely, computing $\mathcal{V}(\mathbf{x})$;
\item updating $\mathbf{x}$ to the vector $\textbf{c}(\mathcal{V}(\mathbf{x}))$.
\end{enumerate}
In mathematical terms, for $k \in \N$, the algorithm is described as
\begin{equation}\label{eq:Lloyd}
\mathbf{x}(k+1)=\textbf{c}(\mathcal{V}(\mathbf{x}(k))).
\end{equation}
It can be shown \cite{Cortes:04} that the function $H(\mathbf{x}; \mu)$ is monotonically non-increasing along the solutions of \eqref{eq:Lloyd} and that all the solutions of \eqref{eq:Lloyd} converge asymptotically to the set of configurations that generate centroidal Voronoi partitions.
It is well known \cite{Cortes:04} that the set of centroidal Voronoi partitions of the pair $\left( \mathcal{X}, \mu \right)$ are the critical points of the coverage function $H(\textbf{x};\mu)$.

\section{Problem Formulation}
\label{PF}
Let $\mu: \mathcal{X} \rightarrow \mathbb{R}$ an unknown function modeled as the realization of a zero-mean Gaussian random field with covariance $K:\mathcal{X} \times \mathcal{X} \rightarrow \mathbb{R}$. We restrict our attention to radial kernels, i.e. $K(a,b)=h(\parallel a - b \parallel)$, such that if $\parallel a - b \parallel \leq \parallel c - d \parallel$ then $h(\parallel a - b \parallel) \leq h(\parallel c - d \parallel)$ and $K(x,x) = \lambda, \ \forall x \in \mathcal{X}$.

Assume we are given a central base-station, and  $N$ robotic agents each moving in the space $\mathcal{X}$. The function $\mu$ is assumed to be unknown to both the agents and the central unit. Each agent $i \in \until{N}$ is required to have the following basic computation, communication and sensing capabilities:
\begin{itemize}
\item[(C1)] agent $i$ can  identify itself to the base station and can send information to the base station;
%\item[(C2)] agent $i$ can store in memory a convex polygon $W_i$ subset of $\mathcal{X}$ and a point $c_i \in W_i$;
\item[(C2)] agent $i$ can sense the function $\mu$ in the position it occupies; specifically, if $x_i$ denotes its current position,  it can take the noisy measurement 
\begin{equation*}
y(x_i) = \mu (x_i) + \nu_i,
\end{equation*}
where $\nu \backsim \mathcal{N}(0,\sigma^2)$, independent of the unknown function $\mu$,  and all mutually independent. 
\end{itemize}
The base station must have the following capabilities
\begin{itemize}
\item[(C3)] it can store all the measurements taken by all the agents; 
\item[(C4)] it can perform computations of partitions of $\mathcal{X}$;
\item[(C5)]it can send information to each robot;
\item[(C6)] it can store an estimate $\hat{\mu}$ of the function $\mu$ and of the posterior variance. 
\end{itemize}
The ultimate goal of the  group of agents and central base-station is twofold:
\begin{enumerate}
\item to explore the environment $\mathcal{X}$ through the agents, namely, to provide an accurate estimate $\hat{\mu}$ of the function $\mu$ exploiting the measurements taken by the agents; 
\item to compute a \emph{good} partitioning of $\mathcal{X}$ using the estimate $\hat{\mu}$,. 
\end{enumerate}

\section{The algorithm}
\label{AL}
To achieve the above goal the following \emph{Estimation + Coverage algorithm} (denoted hereafter as EC algorithm) is employed. \\

\begin{algorithm}
\caption{EC}
\label{alg:EC}
\begin{algorithmic}[1]
\REQUIRE The central base station (CBS) stores in memory all the measurements.
\FOR{k = 1,2,\dots}
\STATE {\bf Measurements collection}: For $i \in \until{N}$, agent $i$ takes the measurement $y_{i,k}$ and sends it to CBS.
\newline
\STATE {\bf Estimate update}: Based on $\bf x_k, x_{k-1}, \ldots, x_0$   and $\left\{y_{1,s}, \ldots, y_{N,s} \right\}_{s=0}^{k}$ CBS computes $\hat{\mu}_k$ and its posterior.
\newline
\STATE {\bf Trajectory update}: Based on $\hat{\mu}_k$ CBS computes $\bf u_k$ and sends it to agents. Agents update position as $\bf x_{k+1}= x_k +u_k$.
\ENDFOR
\end{algorithmic}
\end{algorithm}

Now, introducing the dynamic, we have that for each $k \in \N$ the central base-station stores in memory a partition ${\mathcal W}_k=(W_{1,k}, \ldots, W_{N,k})$ of $\mathcal{X}$, the corresponding list of centroids ${\bf c}_k=(c_{1,k},\ldots,c_{N,k})$, the positions of the robots $(x_{1,k},\ldots, x_{N,k})$ and all the measurements received up to $k$ by the agents.
For $k \in \N$, agent $i$, $i \in \until{N}$, moves according to the following first-order discrete-time dynamics 
$$
x_{i,k+1}=x_{i,k}+u_{i,k}
$$ 
where the input $u_{i,k}$ is assigned to agent $i$ by the central base-station. As soon as agent $i$ reaches the new position $x_{i,k+1}$, it senses the function $\mu$ in $x_{i,k+1}$ taking the measurement $y_{i,k+1} = \mu (x_{i,k+1}) + \nu_{i,k}$ and it sends $y_{i,k+1}$ to the central base-station. The central base-station, based on the new measurements gathered $\left\{y_{i,k+1}\right\}_{i=1}^N$ and on the past measurements, computes a new estimate $\hat{\mu}_{k+1}$ of $\mu$; additionally it updates the partition $\mathcal{W}_k$, setting $\mathcal{W}_{k+1}=\mathcal{V}(x_{1,k+1},\ldots, x_{N,k+1})$.  

The goal is to iteratively update the position of the agents in such a way that, in a suitable metric, $\hat{\mu} \to \mu$ and the Coverage function assumes values as small as possible.

In next subsections we will explain how the central base-station updates the estimate $\hat{\mu}$ based on the measurements collected from the agents, and how it design the control inputs to drive the trajectories of the agents. It is quite intuitive that in order to have a better and better  estimate of the function $\mu$, the measurements  have to be taken to reduce as much as possible a functional of the posterior variance, in particular we will adopt the maximum of the posterior variance. To do so, in the first phase of the EC algorithm the agents will be spurred to explore the environment toward the regions which have been less visited. When the error-covariance of the estimate $\hat{\mu}$ is small enough everywhere, the central base-station will update the agents' position to reduce as much as possible the value of the coverage function.
   
To simplify the notation let us introduce
\begin{equation*}
z_{i,k} = \{ x_{i,k},y_{i,k} \}, \quad i=1,\ldots,N.
\end{equation*}
One of the key aspects of the algorithm is related to the agents movement, 
which establishes how positions $x_{i,k}$ are generated.  In particular, as clear in the sequel, each 
$x_{i,k}$ is a non Markovian process, depending on the whole past history $z_{i,1},\ldots,z_{i,k-1}$, $i=1,\ldots,N$.
It is useful to describe first the function estimator, then detailing the agent dynamics.\\

\subsection{Function estimate and posterior variance}

Hereby, we use $Z_{N,t}$ to denote the set $ \{ z_{i,k} \}$ with $i=1,\ldots,N$ and $k=1,\ldots,t$.
The agents movements are assumed to be regulated by probability densities 
fully defined by $Z_{N,t}$. It comes that the minimum variance estimate of $\mu$ given $Z_{N,t}$ is
%Irrespective of the particular stochastic machinery underlying the agents movement, 
%it is easy to show that the minimum variance estimate of $\mu$ given $Z_{N,t}$ is
\begin{eqnarray} 
&& \hat{\mu}_{t}(x)  = \mathbb{E}\left[\mu(x)| Z_{N,t}  \right] = \sum_{i=1}^N \sum_{k=1}^{t} c_i K(x_{i,k},\cdot) \label{Esmu}
\end{eqnarray}
where
\begin{eqnarray*}
&& \begin{bmatrix}
c_1\\
\vdots\\
c_N
\end{bmatrix}
= (\bar{K} + \sigma^2\mathbb{I})^{-1}
\begin{bmatrix}
y_{1,1}\\
\vdots\\
y_{N,t}
\end{bmatrix}
\end{eqnarray*}
and
\begin{equation*}
\bar{K} = 
\begin{bmatrix}
K(x_{1,1},x_{1,1}) & \ldots &  K(x_{1,1},x_{N,t}) \\
\vdots & & \vdots \\
K(x_{N,t},x_{1,1}) & \ldots &  K(x_{N,t},x_{N,t}) \\
\end{bmatrix}.
\end{equation*}
The a posteriori variance of the estimate, in a generic input location $x \in \mathcal{X}$, is
\begin{eqnarray} 
 V(x) =  \text{Var} \left[ \mu(x) | Z_{N,t} \right] 
= K(x,x) - 
\begin{bmatrix}
K(x_{1,1},x) & \ldots & K(x_{N,t},x)
\end{bmatrix}
(\bar{K} + \sigma^2 \mathbb{I})^{-1}
\begin{bmatrix}
K(x_{1,1},x) \\
\vdots \\
 K(x_{N,t},x)
\end{bmatrix}.
\label{Varmu}
\end{eqnarray}

\subsection{Description of agents dynamics} % (input locations generation)}
\label{Phases}
The generation of the control input can be divided in two phases: in the first, estimation and coverage are carried out together, while, when the estimate is good enough, i.e. the posterior variance is uniformly small, automatically the control switches to the second phase, where the standard coverage control algorithm reviewed in Section \ref{CCA} is deployed. 
\subsubsection{Phase I}
\label{sct:phaseI}
Let 

\begin{equation}
u_{i,k} = 
\begin{bmatrix}
\Re e( \rho_i e^{j\theta_i} )\\
\Im m (\rho_i e^{j\theta_i} )\\
\end{bmatrix}
\quad i=1,\ldots,N
\end{equation}
then the agents dynamics, for $i=1,\ldots,N$, are defined by
\begin{equation*}
x_{i,k+1} =
\begin{cases}
 x_{i,k} + u_{i,k}  & \text{if } x_{i,k} + u_{i,k} \in \mathcal{X} \\
 x_{i,k} & \text{if } x_{i,k} + u_{i,k} \notin \mathcal{X}
\end{cases}
\end{equation*}

Hence, variation of the agent's position is given by the random vector $\rho_i e^{j\theta_i}$,
where $\theta_i$ is a random variable on $[0,2\pi]$, determining the movement's direction,  
while $\rho_i$ is another random variable establishing the step length. The peculiarities
of our approach are the following ones:

\begin{itemize}
\item the statistics of $(\theta_i,\rho_i)$ vary over time and depend on the past history through 
the estimate $\hat{\mu}_t$ and a function $a(\cdot)$ of the maximum of its posterior variance, i.e.
$$
a\left(\max_{x \in \mathcal{X}} V(x) \right).
$$
Note that $a$ varies over time since it depends on the posterior variance
which also varies over time as the agents move over $\mathcal{X}$.
Hereby, to simplify notation, we use $a(t)$ to stress this dependence. 
In this way, at every $t$, a suitable trade-off is established 
between centroids targeting, which are never perfectly known, being function of $\mu$,
and the need of reducing their uncertainty. These two goals are called {\it{exploration}} and {\it{exploitation}} in \cite{Schwager};
\item the probability densities of $\theta_i$ and $\rho_i$ are assumed to be uniformly bounded below.
This means that, irrespective of the particular agent's position and instant $t$, 
there exists $\epsilon>0$ such that every set of Lebesgue measure $\ell>0$
can be reached in one step with probability greater than $\ell \epsilon$. 
\end{itemize}

\begin{example}
\label{exampleOfA}
We provide a concrete example by describing the specific update rule adopted during the numerical experiments
reported in section \ref{sct:NmrRsl}. The random variable $\rho_i$ is a truncated Gaussian, 
constrained to assume positive values, while $\theta_i$ is a bimodal Gaussian with support limited to the interval $[0,2\pi]$. 
More specifically,  for $i=1,\ldots,N$, the density of $\theta_i$ is
\begin{equation*}
p(\theta_i) = 
\begin{cases}
\frac{1-a(t)}{b_i(t)} e^{-\frac{(\theta_i-\theta_{C_i}(t))^2}{\sigma_{C_i}^2}} + \frac{a(t)}{c_i(t)} e^{-\frac{(\theta_i-\theta_{\Delta_i}(t))^2}{\sigma_{\Delta_i}^2}}, \quad & \theta_i \in [0,2\pi] \\
 0, & \theta_i \notin [0,2\pi]
\end{cases}
\end{equation*} 
where
\begin{equation*}
b_i(t) = \int_0^{2\pi} e^{-\frac{(\theta_i-\theta_{C_i}(t))^2}{\sigma_{C_i}^2}} d\theta_i \quad ,\quad c_i(t) = \int_0^{2\pi} e^{-\frac{(\theta_i-\theta_{\Delta_i}(t))^2}{\sigma_{\Delta_i}^2}}d\theta_i 
\end{equation*}
where 
\begin{itemize}
\item $\theta_{C_i}(t)$ determines the direction to follow at instant $t$ to reach the current estimate of the Voronoi centroid of the agent $i$
computed using $\hat{\mu}_t$ as defined in (\ref{Esmu});
\item $\theta_{\Delta_i}$ determines the direction given by the gradient of the posterior variance (\ref{Varmu}) 
computed at the input location occupied by the $i-th$ agent at the instant $t$;
\item $a(t) \in [0,1]$ is a control parameter that establishes the trade-off between exploration and exploitation at instant $t$.
In the next section an automatic way to tune this parameter based on the posterior variance will be presented; %allows to give more weight to the centroid or to the gradient direction;
\item $\sigma_{C_i}^2,\sigma_{\Delta_i}^2$ determine the level of dispersion of the density
around the directions given by $\theta_{C_i}$ and  $\theta_{\Delta_i}$.
\end{itemize}

A simple heuristic that allows to automatically determine the value of $a$ is based on the maximum of the posterior variance, with the constraint that $a$ has to satisfy the following conditions:
\begin{enumerate}
\item $a(t)$ has to be continuos as function of the maximum of the posterior,
\item $a(t)$ has to be monotonically increasing with the maximum of the posterior,
\item if $\max_{x \in \mathcal{X}} V(x) = \lambda$ then $a(t) = 1$ ,
\item  if $\max_{x \in \mathcal{X}} V(x) = 0$  then  $a(t) = 0$. 
\end{enumerate} 
Two examples are reported in Figure \ref{fig:exampleOfA}.

\begin{figure}[h!]
\begin{center}
\begin{tabular}{c}
\includegraphics[width=0.8\columnwidth]{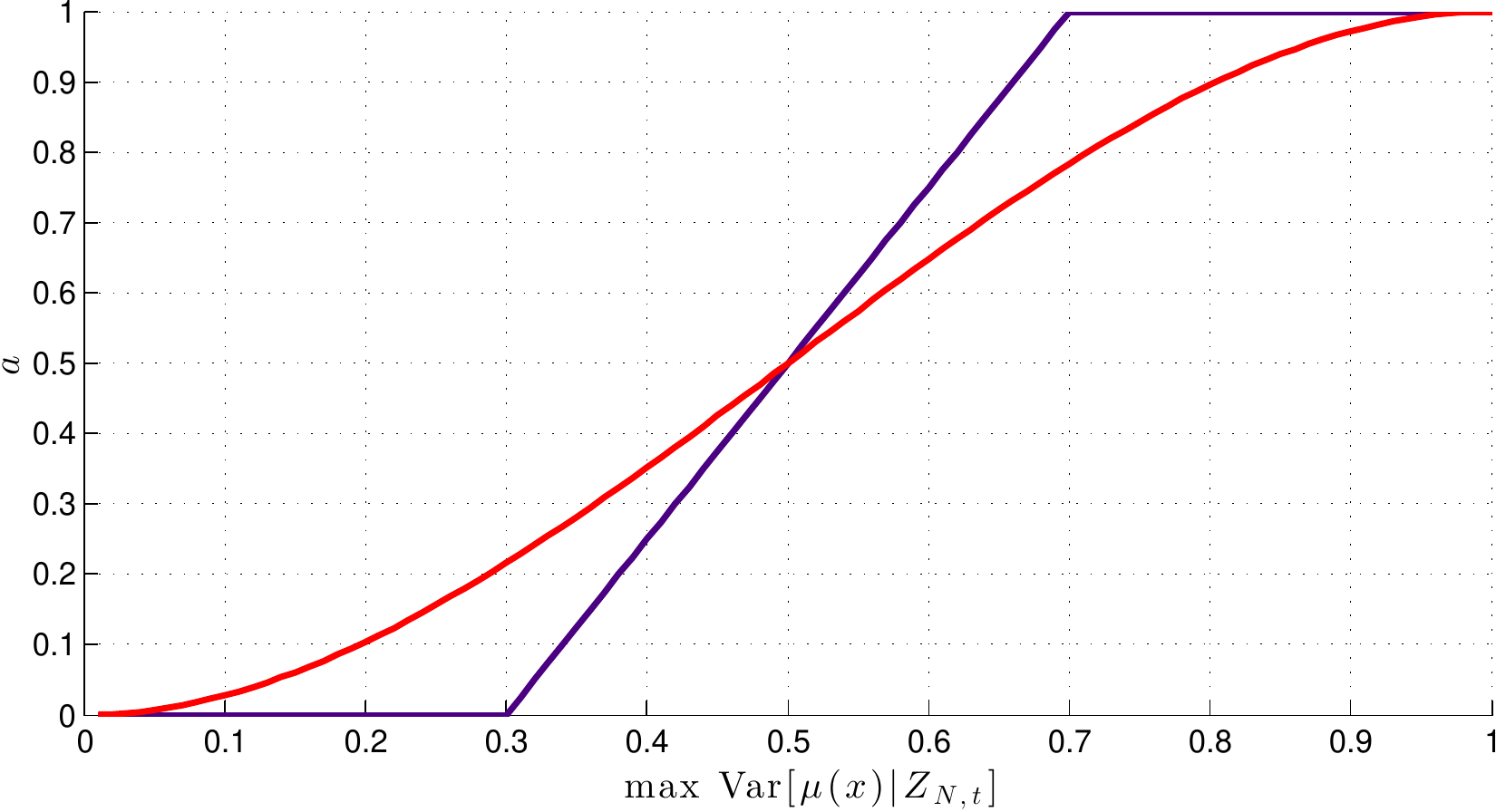}
\end{tabular}
\caption{\label{fig:exampleOfA} Here are reported two examples of a(t). In this case the $\max_{x \in \mathcal{X}} V(x) = \lambda = 1$.}
\end{center}
\end{figure}

At the beginning, being the posterior variance large, $a(t)$ will be close to $1$ and the agents will just explore the domain. Thanks to the monotonicity of $a(t)$, while the maximum of the posterior will be reduced, also $a(t)$ will be reduced and consequently the agents will privilege the coverage.
$\blacksquare$ \\
\end{example}

\subsubsection{Phase II}
When $a(t)$ is under a certain threshold, i.e. the posterior variance is uniformly low, the control input switches from the update rule described in section \ref{sct:phaseI} to $u_{i,k} = -(x_{i,k} - c_{i,k})$, so the agents will directly reach the estimated centroids. In other words, in this phase the Lloyd's algorithm is performed with the unknown function set to the estimate obtained at the end of the first phase. 

\section{Convergence properties of the algorithm}
\label{CP}
It is important to verify that (in probability) the posterior variance can be reduced as much as we want. Indeed, this fact implies that (with probability one) the agents dynamics will switch from phase I to phase II. The following result holds.

\begin{proposition} 
\label{propVariance}
Let $\mu$ be a zero-mean Gaussian random field of radial covariance $K$. Then, 
$\forall \epsilon \geq 0, \forall \delta  \in (0,1] $ there exists $t_0$ such that,
$\forall t \geq t_0$, one has:
\begin{equation*}
Pr\left[\max_{x \in {\mathcal{X}}} \text{ Var } \left( \mu(x)| Z_{N,t} \right) \leq \epsilon \right] \geq 1 - \delta 
\end{equation*}
\end{proposition}
\begin{proof}
Consider the following inequality
\begin{equation*}
\lambda - \frac{(\lambda - \alpha)^2}{\lambda + \frac{\sigma^2}{m}} \leq \epsilon.
\end{equation*}
Then, we can always choose a pair $\bar{\alpha}$ and $\bar{m}$ such that the previous inequality holds. By the continuity of the kernel, there exists a partition, function of $\bar{\alpha}$, given by all the subset $\mathcal{D}_j \subseteq \mathcal{X}$ such that $K(x,x^*) \geq \lambda - \bar{\alpha}, \ \forall x,x^* \in \mathcal{D}_j$.
For a sufficiently large $t$, with a probability greater then $1-\delta$, we can collect $\bar{m}$ or more measurements in each $\mathcal{D}_j$. 
In fact, $\forall A \subseteq \mathcal{X}$ and $\forall x_1 \in \mathcal{X}$, $Pr\left[ x(k+1) \in A|x(k) = x_1 \right] \geq \epsilon \ell_A$, where $\ell_A$ is the Lebesgue measure of $A$, since the probability densities of $\theta_i$ and $\rho_i$ are bounded below.\\ %, and with a probability greater then $\epsilon$ we can reach any input location in $\mathcal{X}$.\\
Now it is not restrictive consider only $\bar{m}$ measurements falling in $\mathcal{D}_j$, which are denoted by $z_{1}^j, \ldots, z_{\bar{m}}^j$
and collected on the input locations $x_{1}^j, \ldots, x_{\bar{m}}^j$. Calling $\bar{K}_j$ the sampled kernel in the input location falling in $\mathcal{D}_{j}$ and thanks to the fact that $Tr(\bar{K}_j) = \sum \Lambda(\bar{K}_j) = m\lambda$ (where $\Lambda(\bar{K}_j)$ is the set of eigenvalues of $\bar{K}_j$) and that all the eigenvalues of $\bar{K}_j$ are real and non negative ($\bar{K}_j$ is symmetric and semi positive definite), it holds that $\bar{K}_j \preceq \bar{m} \lambda \mathbb{I}$ so that 
\begin{equation*}
(\bar{K}_j + \sigma^2) \preceq (\bar{m} \lambda + \sigma^2)\mathbb{I} \Rightarrow (\bar{K}_j + \sigma^2)^{-1} \succeq (\bar{m} \lambda + \sigma^2)^{-1}\mathbb{I}.
\end{equation*}
So with probability greater then $1-\delta$ it is true that
\begin{equation*}
\begin{split}
 \text{Var} \left[ \mu(x) | z_{i}^j, \ldots, z_{\bar{m}}^j \right] 
&= K(x,x) - 
\begin{bmatrix}
K(x_{1}^j,x) & \ldots & K(x_{\bar{m}}^j,x)
\end{bmatrix}
(\bar{K}_j + \sigma^2 \mathbb{I})^{-1}
\begin{bmatrix}
K(x_{1}^j,x) \\
\vdots \\
 K(x_{\bar{m}}^j,x)
\end{bmatrix}\\
&  \leq \lambda - \frac{\sum_{h=1}^{\bar{m}} K(x_{h}^j,x)^2 }{\bar{m}\lambda + \sigma^2} \leq \lambda - \frac{\bar{m}(\lambda - \bar{\alpha})^2 }{\bar{m}\lambda + \sigma^2} = \lambda - \frac{(\lambda - \bar{\alpha})^2}{\lambda + \frac{\sigma^2}{\bar{m}}} \leq \epsilon
\end{split}
\end{equation*}
thus proving the statement.
\end{proof}

The consequence of Proposition \ref{propVariance} is that with probability one there exists a time $\bar{k}$ such that the agents dynamics switch from phase I to phase II, namely the agents dynamics will be rule by 
\begin{equation}
\label{phaseII}
\textbf{x}_{k+1} = \textbf{c}\left( \mathcal{V}(\textbf{x}(k)) \right)
\end{equation}
for $k > \bar{k}$, where the centroids are computed according to the estimate $\hat{\mu}_{\bar{k}}$.

\begin{proposition}
The trajectory generated by \ref{phaseII} converges to the set of configurations that generate centroidal Voronoi partitions of the pair $\left( \mathcal{X}, \hat{\mu}_{\bar{k}} \right)$.
\end{proposition}

\section{Numerical Results}
\label{sct:NmrRsl}
In this section, we provide some simulations implementing the new estimation and coverage algorithm. We consider a team of $N=8$ agents placed, with a random initial position, in the domain $\mathcal{X}=[0,1] \times [0,1]$. Moreover, we use the Gaussian kernel
\begin{equation*}
K(x,x') = e^{-\frac{\|x-x'\|^2}{0.02}}
\end{equation*}
with the estimator and the posterior variance given by (\ref{Esmu}) and (\ref{Varmu}), respectively.
The unknown sensory function $\mu$ is a combination of four bi-dimensional Gaussian:
\begin{equation*}
\mu(x) = 
%\begin{bmatrix}
%20\\
%20
%\end{bmatrix}
20 \left( e^{\frac{-\|x - \mu_1\|^2}{0.04}} 
+ e^{\frac{\|x - \mu_2\|^2}{0.04}} \right)+
%\begin{bmatrix}
%5\\
%5
%\end{bmatrix}
5\left( e^{-\frac{\|x - \mu_3\|^2}{0.04}} 
+ e^{-\frac{\|x - \mu_4\|^2}{0.04}} \right),
\end{equation*}
where 
\begin{equation*}
\mu_1 = 
\begin{bmatrix}
0.2\\
0.2
\end{bmatrix}
\quad
\mu_2 = 
\begin{bmatrix}
0.8\\
0.8
\end{bmatrix}
\quad
\mu_3 = 
\begin{bmatrix}
0.8\\
0.2
\end{bmatrix}
\quad
\mu_4 = 
\begin{bmatrix}
0.2\\
0.8
\end{bmatrix}.
\end{equation*}
For computational reasons, the function $\mu$ and the posterior variance are evaluated over a grid of step $0.05$. 
The two parameters $\sigma_{\Delta_i}^2$ and $\sigma_{C_i}^2$ are both set to $0.1$ and the threshold that allows to switch from phase I to phase II is equal to $0.3$.\\
The adopted $a(t)$ is as described in Example \ref{exampleOfA}. This means that, when the maximum of the posterior is large, the value of $a(t)$ is also large to allow a good estimation. Instead, when the maximum of the posterior variance becomes small, also the value of $a(t)$ is reduced to 
favor agents movement towards the centroids. \\
An example is in Figure \ref{fig:Gradient} which displays the posterior variance (contour plot), the gradient (quiver plot) and the agents (red diamonds). The figure illustrates results from the first iteration (just because the plot is more clear).
\begin{figure}[h!]
\begin{center}
\begin{tabular}{c}
\includegraphics[width=0.8\columnwidth]{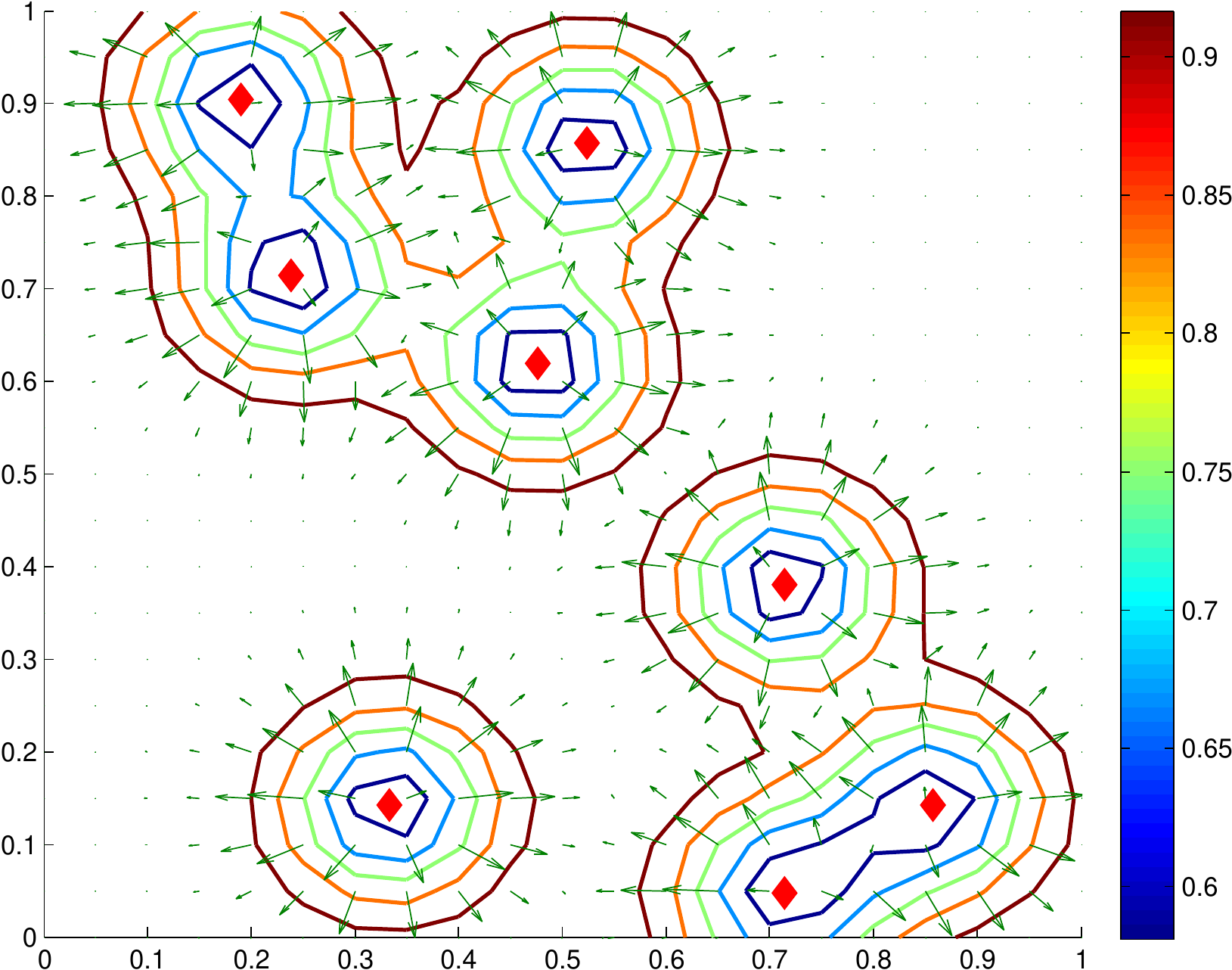}
\end{tabular}
\caption{\label{fig:Gradient}Contour plot of the posterior variance with the directions of the gradient. The red diamonds are the agents.} 
\end{center}
\end{figure}
\\Figure \ref{fig:PosteriorEvolution} plots the profile of the maximum, the average and the minimum of the posterior, as a function of the number
of iterations.
\begin{figure}[h!]
\begin{center}
\begin{tabular}{c}
\includegraphics[width=0.8\columnwidth]{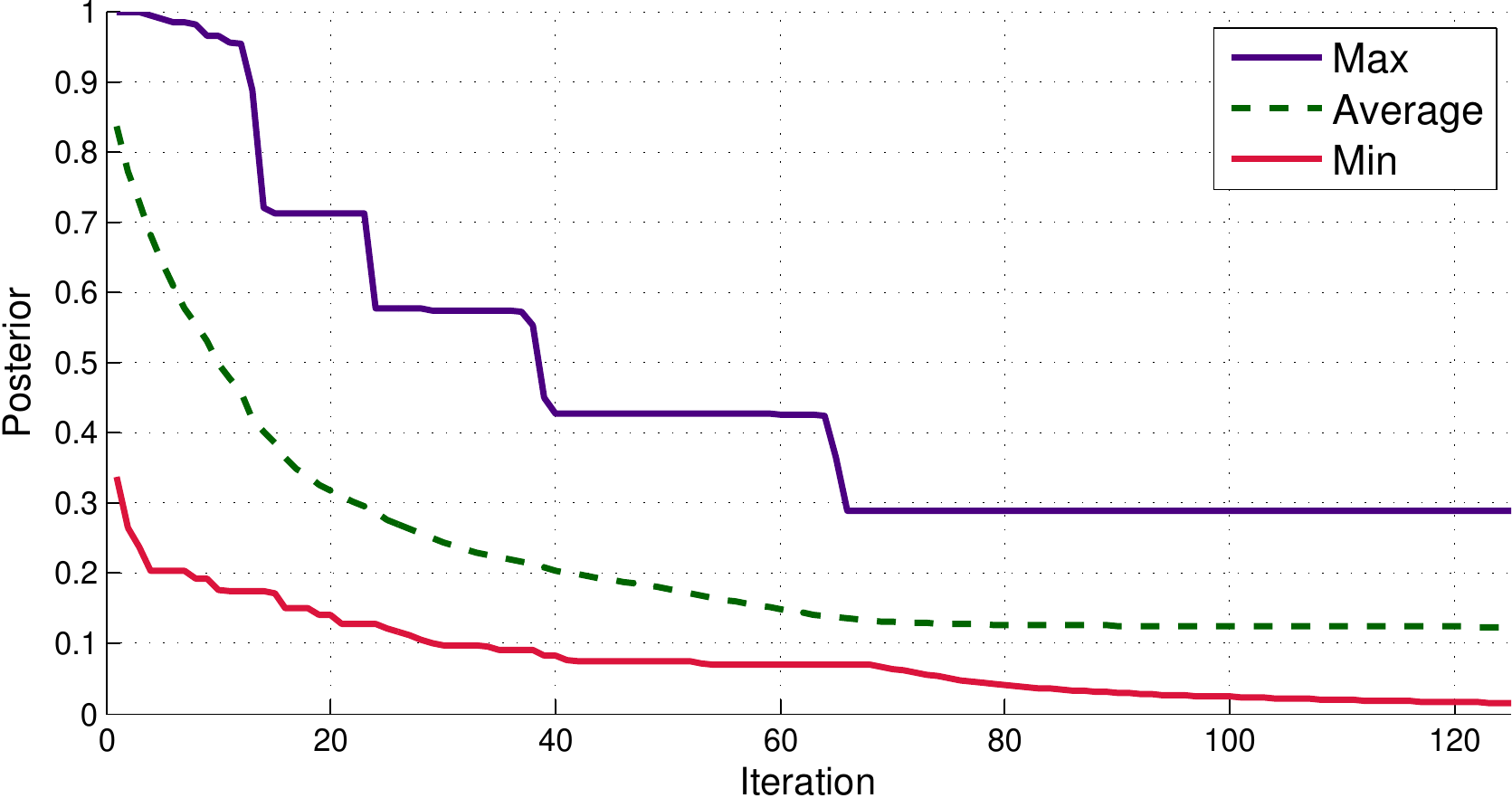}
\end{tabular}
\caption{\label{fig:PosteriorEvolution}Evolution of the maximum, the average $\left(\frac{1}{\int_{\mathcal{X}}dx} \int_{\mathcal{X}} V(x) dx \right)$ and the minimum of the posterior over the time.} 
\end{center}
\end{figure}
Finally, Figure \ref{fig:Centroid} reports the Voronoi regions associated with the agents, as well as the estimated function $\hat{\mu}(x)$ (contour plot). 
The final agents positions (red circles) are close to the ideal agents positions, computed using the true sensory function $\mu$ (black circles).%  One can see that agents final positions are satisfactory. 
\begin{figure}[h!]
\begin{center}
\begin{tabular}{c}
\includegraphics[width=0.8\columnwidth]{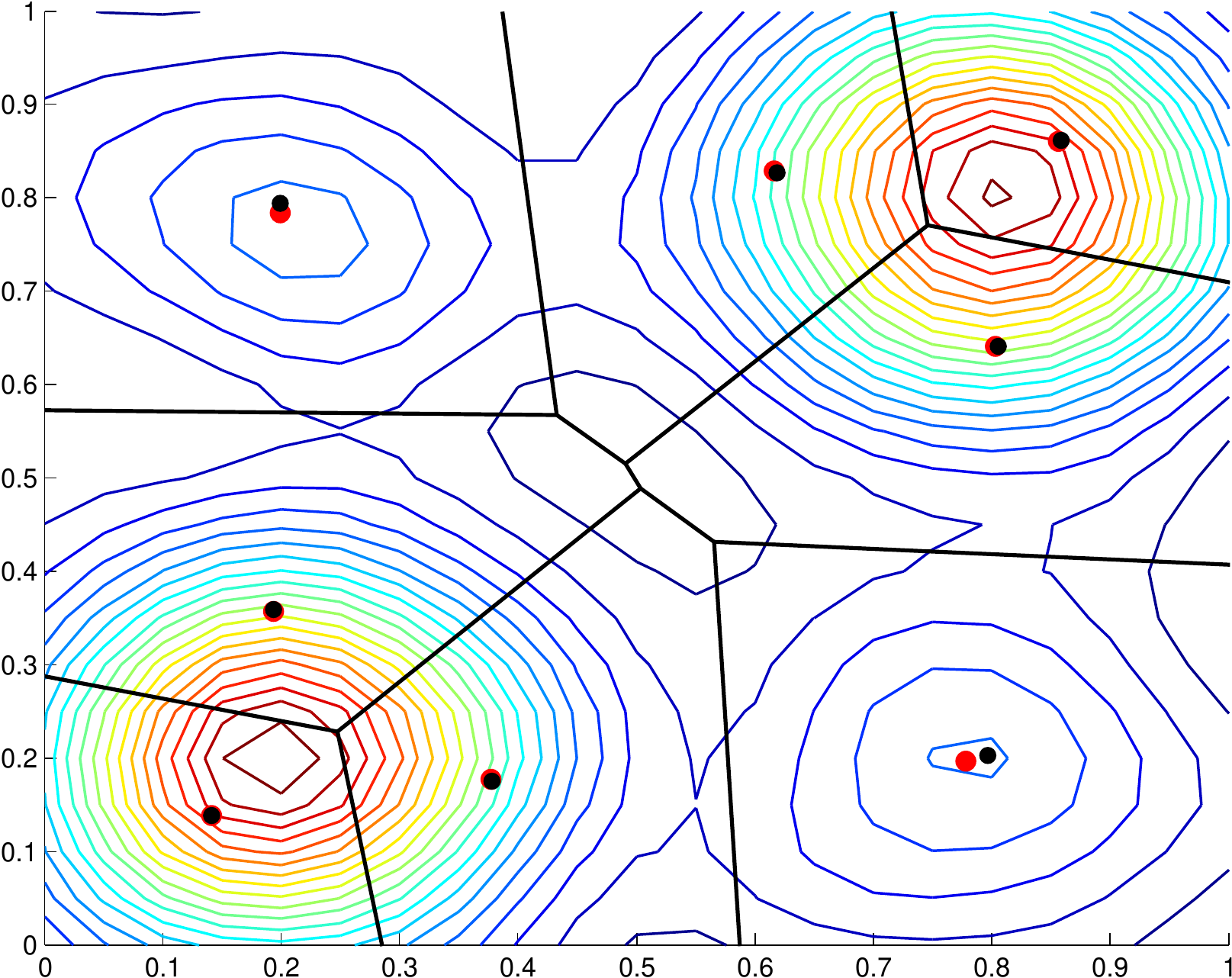}
\end{tabular}
\label{fig:Centroid}
\caption{Voronoi regions of the agents with contour plot of the density function. The red circles are the agents positions, the black circles are the centroids positions computed using $\mu$ instead of $\hat{\mu}_{\bar{k}}$.} 
\end{center}
\end{figure}

\section{Conclusions}

We have proposed a new algorithm to perform simultaneously estimation and coverage.
The sensory function is seen as a Gaussian random field which has to be reconstructed in an on line manner.
A set of control inputs establish the agents movement, trying to balance coverage and estimation.
We have seen that the resulting problem is also an instance of a non standard function learning problem
where input locations follow a non Markovian process with stochastic adaption allowed to happen infinitely often.
Convergence of the estimator has been discussed also assuming that the function prior is not correct (see Appendix).
Numerical experiments show good performance. 
Even if the centralized algorithm finds many applications in different fields, such as \cite{PereiraCommunication} and \cite{Shah}, we are also working on a distributed version. The core of the algorithm is based on the on-line non-parametric regression studied  in \cite{Varagnolo}.
In addition, we also plan to provide a distributed algorithm
possibly also accounting for time variance of the sensory function.

\section*{Appendix: convergence in RKHS}

In section \ref{Phases} we have seen that the proposed coverage algorithm is divided into two phases.
In the first phase, a trade-off between estimation and coverage is searched for at every step, while 
the final coverage step is performed in the second one which starts when the posterior variance
is under a certain threshold. Recall that Proposition \ref{propVariance} ensures (in probability) that the second
phase will always be reached. This result is obtained under the assumptions that the Bayesian prior on the unknown function is correct.
In this Appendix we will relax this assumption, just assuming that the function belongs to the reproducing kernel Hilbert space
induced by the covariance $K$ which is thus seen as a reproducing kernel. 
Our key question is to asses if, in the first phase where input locations follow a very complex 
non Markovian process and adaptation may happen infinitely often, the estimate is converging to the true function
in some norm. We will se that the answer is positive: under some technical conditions
convergence (in probability) holds under the RKHS norm (which also implies convergence in the sup-norm).

Dependence of input locations
and measurements on the agent number is 
skipped to simplify notation. Hence, 
the set of input locations explored by the network and related meaurements
available at instant $t$ are denoted by $x_1,\ldots,x_t$ and $y_1,\ldots,y_t$, respectively.\\
The following proposition relies on the well known relationship between 
Bayes estimation of Gaussian random fields and regularization in RKHS.

\begin{proposition}
Let $\H$ be the RKHS induced by the kernel $K: \mathcal{X} \times \mathcal{X} \rightarrow \R$,
with norm denoted by $\| \cdot \|_{\H}$. Then,  
for any $x \in \X$, one has 
\begin{equation}\label{Tikhonov}
\hat{\mu}_{t}  = \arg \min_{f \in \H}  J(f) \\
\end{equation}
where
\begin{equation}\label{J}
J(f) = \frac{\sum_{k=1}^t \left( y_i- f(x_i) \right)^2}{t}     +   \gamma \| f \|_{\H}^2, \quad \gamma=\frac{\sigma^2}{t} 
\end{equation}
\end{proposition}
\begin{flushright}
$\blacksquare$
\end{flushright}

We consider a very general framework to describe the process $x_i$, 
which contains that previously described as special case.
The input locations  $x_i$ are thought of as random vectors each randomly drawn from a Borel nondegenerate probability density function $p_i \in \mathcal{P}$, with
the noise $\nu_t$ independent of $x_1,\ldots,x_t$ for any $t$.
We do not specify any particular stochastic or deterministic mechanism through which
the $p_i$ evolve over time. We just need two conditions regarding 
the behavior of some covariances and the smoothness of $\mu$, as 
detailed in the next subsection. 

\subsection{Assumptions}
\label{Assum}

To state our assumptions, first we need to set up some additional notation.
We use $co \mathcal{P}$ to denote the convex hull of $\mathcal{P}$,
i.e. the smallest convex set of densities containing $\mathcal{P}$.
Let also $\L^2_p$ be the Lebesque space parametrized 
by the density $p$, i.e. the space of real functions $f:X \rightarrow \R$ such that 
$$
\| f \|^2_p := \int_X f^2(a) p(a) da < \infty.
$$

Our first assumption regards the decay of the 
 covariance between a class of functions evaluated at different input locations.\\
% {\bf{Va dimostrato che, se le densit\'a $p_i$ sono tutte uniformly bounded below,
% questa assunzione \'e automaticamente soddisfatta.}}
 
\begin{assumption}[covariances decay] \label{A1}
%Let $B_q$ the following set of functions
%$$
%B_q= \{  f :    \}
%$$
Let $f_1,f_2$ be any couple of functions satisfying
$$
\| f_{1 \vee 2} \|_p < q < \infty, \quad \forall p \in co\mathcal{P}
$$
Then, for every time instant $i$, there exists a constant  $\A_1< \infty$, dependent on $q$ but not on $f$, such that
\begin{equation*}
%\sum_{k=0}^\infty  \left| Cov(f_1(x_i)+f_2(x_i)\nu_i,f_1(x_{i+k})+f_2(x_{i+k})\nu_{i+k}) \right| < \A_1 
\left| \sum_{k=0}^\infty   Cov(f_1(x_i)+f_2(x_i)\nu_i,f_1(x_{i+k})) \right| < \A_1 
\end{equation*}
\end{assumption}
\begin{flushright}
$\blacksquare$
\end{flushright}

%\begin{remark}
%A natural question is if the stochastic process generated for exploring $X$
%satisfies the above assumption.
%A formal proof would require
%a non trivial analysis of non Markovian processes over uncountable state spaces. 
%Simulations however suggest that covariances converge
%to zero at an exponential rate. This fact is not surprising.
%In fact, our transition kernels are all irreducible and define
%a perturbed version of a random walk Metropolis-Hastings scheme which is  
%known to be geometrically ergodic when $X$ is a compact [].  
%To further support this fact, think also of a discretization of $X$
%in cells, so that transition kernels become irreducible stochastic matrices,
%all associated with particular random walks on a graph with uniform distribution as invariant distribution. According to Theorem 2 in [], the product  of any deterministic sequence of such transition matrices converges geometrically to the uniform distribution so that the input locations
%sequence satisfies Assumption 1.
%\end{remark}

The second assumption is related to smoothness of $\mu$.
Below, given a density $p$, the operator
$L_p : \L^2_p \rightarrow \H$ is defined by 
\begin{equation}\label{BarL}
L_p[f](x) = \int_{X} K(x,a) f(a) p(a) da, \quad x \in X
\end{equation}

\begin{assumption}[smoothness of the target function] \label{A2}
There exist constants $r$, with $\frac{1}{2}<r \leq 1$, and $\A_2 < \infty$, such that 
\begin{equation}\label{A2Cond}
\sup_{p \in co\mathcal{P}} \| L_p^{-r} \mu  \|_{p} < \A_2.  %\qquad \forall \theta \in \Theta
\end{equation}
%for every possible set of $m$ densities which can be selected from $\mathcal{P}$.
\end{assumption} 
\begin{flushright}
$\blacksquare$
\end{flushright}

\subsection{Consistency in RKHS}
\label{Const}

Our main result is reported below.

\begin{proposition}\label{Main}
Let Assumptions \ref{A1} and \ref{A2} hold.  In addition, let the 
regularization parameter $\gamma$ depend on instant $t$ as follows
\begin{equation}\label{gamman}
\gamma \propto t^{-\alpha}, \quad 0<\alpha < \frac{1}{2}. %n^{-\frac{1}{1+2r}}
\end{equation}
Then, as $t$ goes to infinity, one has
\begin{equation}\label{Result2}
\sup_{x \in X} |\hat{\mu}_{t} (x) - \mu(x) | \longrightarrow_p 0
\end{equation}
where $\longrightarrow_p$ denotes convergence in probability.
\end{proposition}

\begin{proof}
We show that, as $t$ goes to $\infty$, 
the estimator $\hat{\mu}_{t}$ in (\ref{Tikhonov}) 
converges in probability to $\mu$ in the topology of $\H$
and, hence, in that of the continuous functions.
First, some useful notation is introduced.\\ 
Denote by
$$
p_1,\ldots,p_t %, \quad p_i \in \mathcal{P},
$$
the first $t$ densities selected from $\mathcal{P}$
during the first $t$ exploration steps. Note that
repetitions can of course be present, e.g. one can have $p_1=p_2$.
The average density is
\begin{equation}\label{Barp}
\bar{p}_t(x) = \frac{ \sum_{i=1}^t p_i (x) }{t}.
\end{equation}
%As clear in the sequel, 
%such a density plays a key role in the analysis. 
It is useful to indicate with  
$\L^2_t$ the Lebesque space of real functions with norm
$$
\| f \|^2_t := \int_X f^2(a)\bar{p}_t(a) da < \infty.
$$
Note that, in the description of the space and its norm, the integer $t$ in the
subscript replaces $\bar{p}_t$. Following this convention, let also  
\begin{equation}\label{BarL}
L_t[f](x) = \int_{X} K(x,a) f(a) \bar{p}_t(a) da, \quad x \in X
\end{equation}

%%%%%%%%%FINE INTRO %%%%%%%%%%%%%%%%%%

The following function plays a key role in the subsequent analysis:
\begin{equation}\label{fgm}
\bar{\mu}_{t} = \arg \min_{f \in \H}  \| f-\mu \|^2_{\L^2_t} + \gamma \| f \|^2_{\H} 
\end{equation}
Note that, differently from the data-free limit
function introduced in  %$f_{\gamma}$ 
\cite[eq. 2.1]{Smale2007}, here $\bar{\mu}_{t}$
is a (possibly random) time-varying function, depending on the time instant $t$.\\
One has 
\begin{equation}\label{DecomErr}
\|\hat{\mu}_{t} -\mu \|_{\H} \leq  \left \| \bar{\mu}_{t} -\mu \right \|_{\H} + \left \|\hat{\mu}_{t} -\bar{\mu}_{t} \right \|_{\H}%\|\hat{\mu}_{t} -\mu \|_{\H} \leq  \| f_{\theta}-\mu \|_{\H} + \|\hat{\mu}_{t} -f_{\theta} \|_{\H}
\end{equation}
We start analyzing the first term on the RHS of (\ref{DecomErr}).
The average density $\bar{p}_t$ varies over time but never escapes from 
$co \mathcal{P}$.  Then, combining
Assumption \ref{A2} and eq. (3.11) in \cite{Smale2007},
one obtains the following bound uniform in $t$:
\begin{equation}\label{FirstTerm}
\left \| \bar{\mu}_{t} -\mu \right \|_{\H} 
\leq    \gamma^{r-\frac{1}{2}} \| L_t^{-r} \mu  \|_{t} \leq \gamma^{r-\frac{1}{2}} \A_2
\end{equation}

Now, we study $\E \left \|\hat{\mu}_{t} -\bar{\mu}_{t} \right \|_{\H}$, 
i.e. the expectation of the second term on the RHS of (\ref{DecomErr}).
Despite the complex nature of $\bar{\mu}_{t}$, we can apply the same arguments introduced in 
the first part of Section 2 of \cite{Smale2007} which, combined with 
definitions (\ref{Barp},\ref{BarL}), lead to  the equalities
\begin{equation*}%\label{mtheta}
\gamma \bar{\mu}_{t} = L_t[\mu - \bar{\mu}_{t}] =\frac{1}{t}  \sum_{i=1}^{t} L_{p_i} [\mu - \bar{\mu}_{t}]    
\end{equation*}
as well as to the following inequality
\begin{subequations}
\begin{align*}  \small
 & \E \left \|\hat{\mu}_{t} -\bar{\mu}_{t} \right \|_{\H}\\ \nonumber 
 & \leq \frac{1}{\gamma} \small \E  \left[  \left\| \frac{1}{t}  \sum_{i=1}^{t} \left((y_i- \bar{\mu}_{t}(x_i) ) K(x_i, \cdot)   
- L_{p_i}[ \mu - \bar{\mu}_{t}](\cdot)  \right )   \right\|_{\H} \right] 
\end{align*}
\end{subequations}

To gain further insight on the above expression, first consider
\begin{equation}\label{SqNorm}
\E  \left[  \left\| \frac{1}{t}  \sum_{i=1}^{t} \left((y_i- \bar{\mu}_{t}(x_i) ) K(x_i, \cdot)  
- L_{p_i}[ \mu - \bar{\mu}_{t}](\cdot)  \right )   \right\|_{\H}^2 \right] 
\end{equation}
Using the Mercer theorem, we can always find real and positive eigenvalues $\lambda_j$ and
related eigenfunctions $\phi_j$, e.g. orthonormal w.r.t. the classical Lebesgue measure on $X$, such that %expand the kernel $K$ as follows
$$
K(x_i,\cdot)= \sum_{j=1}^{\infty} \lambda_j \phi_j(x_i) \phi_j(\cdot)
$$
Then, one has
\begin{eqnarray*}
(y_i- \bar{\mu}_{t}(x_i) ) K(x_i, \cdot) &=& \sum_{j=1}^{\infty} (y_i- \bar{\mu}_{t}(x_i) ) \lambda_j \phi_j(x_i) \phi_j(\cdot) 
=   \sum_{j=1}^{\infty} a_j(x_i)   \lambda_j  \phi_j(\cdot) 
\end{eqnarray*}
where we have used the following correspondence
$$
a_j(x_i) = (\mu(x_i)+ \nu_i - \bar{\mu}_{t}(x_i) )\phi_j(x_i)%, \qquad b_j(x_i) = \nu_i \phi_j(x_i).
$$
Now, simple calculations show that
\begin{eqnarray*}
\E\left[(y_i- \bar{\mu}_{t}(x_i) ) K(x_i, \cdot)\right] &=& L_{p_i}[ \mu - \bar{\mu}_{t}](\cdot)
=\sum_{j=1}^\infty \E[a_j(x_i)] \lambda_j \phi_j (\cdot)
\end{eqnarray*}

Using RKHS norm's structure, (\ref{SqNorm}) can now be rewritten as follows 
\begin{subequations}
\begin{align*}  \small
&\frac{1}{t^2} \E  \left\|	 \sum_{i=1}^t \sum_{j=1}^\infty (a_j(x_i)-\E[a_j(x_i)])\lambda_j \phi_j(\cdot)	\right\|_{\H}^2  
= \E \left[ 	\sum_{i=1}^t \sum_{k=1}^t \sum_{j=1}^\infty \frac{\lambda_j}{t^2} \left(a_j(x_i) -\E[a_j(x_i)]\right)\left(a_j(x_k)-\E[a_j(x_k)]\right) \right]%\\
%& + \frac{\sigma^2}{t^2} \sum_{i=1}^t \E[K(x_i,x_i)] 
\end{align*}
\end{subequations}
We now obtain an upper bound on the first term present in the rhs of the above equation.
First, taking $f=0$ in the objective in (\ref{fgm}), one has  
$$
\| \bar{\mu}_{t} - \mu \|_t \leq \| \mu \|_t \leq  \sup_{p \in co\mathcal{P}}   \| \mu \|_p < \ell_1 <\infty.
$$
where the last inequality derives from continuity of the function $\mu$ 
on the compact $X$. 
In addition, $\phi_j$ are all contained in a ball of the space of continuous functions, say of radius 
$\ell_2$.\footnote{This holds for the Gaussian kernel and, in practice, also for every covariance adopted in the literature, e.g. spline, 
Laplacian and polynomial kernels.}
This leads to the following bound, uniform in $m$ and $j$: 
$$
\| a_j \|_t \leq \ell_1 \ell_2.
$$
The above inequality permits to exploit Assumption \ref{A1} to obtain
\begin{subequations}
\begin{align*}  \small
& \frac{1}{t^2} \E \left[ 	\sum_{i=1}^t \sum_{k=1}^t \sum_{j=1}^\infty  \lambda_j \left(a_j(x_i) -\E[a_j(x_i)]\right)\left(a_j(x_k)-\E[a_j(x_k)]\right) \right] \\
& \leq  \sum_{j=1}^\infty \frac{\lambda_j}{t^2} \sum_{i=1}^t \sum_{k=1}^t  \left | \E \left[ \left(a_j(x_i) -\E[a_j(x_i)]\right)\left(a_j(x_k)-\E[a_j(x_k)]\right) \right] \right| 
% && \leq \frac{2}{t} \sum_{j=1}^\infty \lambda_j  \sum_{k=1}^\infty \left | \E \left[ \left(a_j(x_i) -\E[a_j(x_i)]\right)\left(a_j(x_k)-\E[a_j(x_k)]\right) \right] \right|
 \leq \frac{2 \A_1 \sum_{j=1}^{\infty} \lambda_j}{t}  %\leq \frac{\A^2_3}{t} 
\end{align*}
\end{subequations}
where, in virtue of the Mercer theorem, $\sum_{j=1}^{\infty} \lambda_j< \infty$ 
(recall, in fact, that each $L_p$ induced by $K$, with $p$ e.g. the uniform distribution on $X$, is a trace class operator).
This last result, together with the Jensen's inequality, leads to % the following bound for (\ref{Etai4})
\begin{equation}\label{Etai7}
 \E[\|\hat{\mu}_{t} -\bar{\mu}_{t} \|_{\H}  ]  \leq   \frac{\sqrt{2 \A_1 \sum_{j=1}^{\infty} \lambda_j }}{\gamma\sqrt{ m} }. 
 %+ \frac{\sigma^2}{t} \sum_{i=1}^t \E[K(x_i,x_i)] }}{\gamma\sqrt{ m} }. 
 \end{equation}
Combining (\ref{Etai7})  with  (\ref{DecomErr}) and (\ref{FirstTerm}), we obtain
\begin{eqnarray*} 
 \E[\|\hat{\mu}_{t}-\mu \|_{\H}] \leq \gamma^{r-\frac{1}{2}} \A_2 +  \frac{\sqrt{2 \A_1 \sum_{j=1}^{\infty} \lambda_j } }{\gamma\sqrt{ m} }
 %+\frac{\sigma^2}{t} \sum_{i=1}^t \E[K(x_i,x_i)] } }{\gamma\sqrt{ m} } 
\end{eqnarray*} 
and this completes the proof.
\end{proof}
\bibliographystyle{IEEEtran}
\bibliography{biblio}

\end{document}